\documentclass[11pt]{article}
\usepackage{amsmath,amssymb,amsthm,mathtools}
\usepackage{geometry}
\usepackage{hyperref}
\usepackage{enumitem}
\usepackage{microtype}
\usepackage{mathrsfs}
\usepackage{bbm}
\title{Existence of Pure Strategy Nash Equilibria in Finite Noncooperative Games}
\author{Shravan Luckraz}
\author{Shravan Luckraz\\
School of Economics and CeDEx China, University of Nottingham Ningbo China\\
\texttt{shravan.luckraz@nottingham.edu.cn}}
\date{July 2026}

\newtheorem{theorem}{Theorem}[section]
\newtheorem{lemma}[theorem]{Lemma}
\newtheorem{proposition}[theorem]{Proposition}
\newtheorem{corollary}[theorem]{Corollary}
\newtheorem{definition}[theorem]{Definition}
\newtheorem{remark}[theorem]{Remark}

\begin{document}

\maketitle

\begin{abstract}
The classical existence result of Nash guarantees that every finite noncooperative game admits an equilibrium in mixed strategies, but it leaves open the question of when pure strategy equilibria exist. This paper develops a structural approach to that question by exploiting properties of the best-response correspondence on finite strategy sets. Building on recent work, we derive new sufficient conditions for the existence of pure strategy Nash equilibria in finite games. We introduce several broad classes of finite games for which pure equilibria are guaranteed, including a class of generalized potential games that generalizes unilaterally competitive games and a class characterized by the existence of an aggregate-payoff maximizer over an ordered set. Our results clarify the role of  acyclicity, and aggregation in producing pure equilibria and connect disparate sufficient-condition results in the literature into a unified framework.
\end{abstract}

\section{Introduction}

\subsection{Background and motivation}
The foundational contribution of Nash (1950, 1951) established that every finite noncooperative game admits at least one Nash equilibrium when players are allowed to randomize over their finite action sets. Nash's proof, by embedding the finite game into its mixed-extension and applying a Kakutani-type fixed-point argument, remains central to game theory. However, the mixed-strategy existence theorem does not address whether equilibria can be realized in pure strategies; that is, whether there exist action profiles composed solely of deterministic choices that are mutual best responses. The distinction is not merely technical: pure strategy equilibria are often more interpretable, easier to compute, and more relevant for applications where randomization is implausible or costly.

Despite extensive work on equilibrium existence in broader settings (for example, continuous strategy spaces, discontinuous payoffs, and infinite-player models) characterizations that are both necessary and sufficient for pure strategy equilibria in finite games remain scarce. The absence of a natural topology on finite strategy sets complicates the direct application of classical fixed-point theorems and topological methods (Fan 1952; Tarski 1955; Abian 1968). Consequently, researchers have pursued alternative structural conditions like monotonicity, supermodularity, acyclicity, and aggregation properties that can guarantee pure equilibria without relying on continuity or convexity assumptions (Topkis 1998; Vives 1990; Monderer and Shapley 1996; Rosenthal 1973).

This paper addresses the gap by developing a framework that uses the combinatorial and order-theoretic structure of best-response correspondences on finite strategy profiles to derive new sufficient conditions for the existence of pure strategy Nash equilibria. Our approach generalizes and extends several strands of the literature, clarifying when and why pure equilibria arise in finite strategic settings (Luckraz 2014; Luckraz 2022a; Luckraz 2022b; Luckraz 2025).

\subsection{Literature review}
A large body of work has established sufficient conditions for pure strategy equilibria in various classes of games. Rosenthal (1973) identified potential games as a rich source of pure equilibria, showing that the existence of a potential function yields pure Nash equilibria via local maximizers. Monderer and Shapley (1996) further developed the theory of potential games, linking potential functions to convergence properties of learning dynamics. Supermodular games and strategic complementarities, studied by Topkis (1998) and Vives (1990), provide another robust route to pure equilibria: monotone best-response maps on lattices admit fixed points by Tarski's theorem (Tarski 1955), and these fixed points correspond to pure Nash equilibria.

Other researchers have relaxed continuity and convexity assumptions to obtain existence results in discontinuous or nonconvex environments. Reny (1999) and Tian (2009) provided existence results under weaker payoff regularity conditions, while Kukushkin (1994, 2004, 2007) examined best-response dynamics and aggregation properties in finite and discrete settings. Abian's fixed-point considerations (Abian 1968) and subsequent refinements (Luckraz 2014) have clarified the role of order and isotonicity assumptions in discrete fixed-point arguments.

More recent work has focused on classes of finite games that are not covered by classical potential or supermodular frameworks. Duersch, Oechssler, and Schipper (2012a, 2012b) studied symmetric two-player zero-sum and imitation dynamics, deriving sufficient conditions for pure equilibria in quasi-concave settings. Research by Iimura and Watanabe (2014), Iimura and Watanabe (2016), Iimura, Maruta, and Watanabe (2019) and Iimura (2020) have advanced the theory of unilaterally competitive and weakly unilaterally competitive games, showing that such strategic structures often guarantee pure equilibria and, in some cases, that the guarantee extends to games with three or more players (Iimura 2020). Amir and De Castro (2017) and Fabrikant et al.\ (2013) have contributed to understanding quasi-monotonic best responses and weak acyclicity, respectively, further enriching the catalog of conditions that produce pure equilibria.

Despite these advances, the literature lacks a unified, fully general characterization of pure strategy Nash equilibria for finite games. Existing results are typically sufficient conditions tailored to particular structural assumptions like potentiality, supermodularity, unilaterally competitive behavior, or aggregation, without providing a single set of necessary and sufficient criteria that applies across finite games. This paper aims to partially fill that gap by focusing directly on the best-response correspondence as a combinatorial object and deriving conditions that are sufficient for the existence of pure equilibria.

\subsection{Challenges in finite strategy spaces}
Two interrelated challenges make the pure-equilibrium existence problem in finite games difficult. First, finite strategy sets lack the topological and convex structure that underpins many classical fixed-point theorems. Without a continuum of strategies, one cannot directly invoke Brouwer or Kakutani arguments on the original strategy space; instead, one must either pass to mixed strategies (as Nash did) or exploit discrete order and lattice structures when they exist (Tarski 1955; Topkis 1998). Second, the combinatorial complexity of best-response correspondences in finite games can be substantial: best-response sets may be multi-valued, non-monotone, and interdependent across players in ways that preclude simple aggregation.

To overcome these obstacles, our analysis treats the best-response map as a discrete multivalued mapping on the finite product of players' action sets and studies its fixed points using order-theoretic and combinatorial tools. This perspective allows us to identify structural features such as acyclicity, monotone selection properties, and the existence of global maximizers of suitably defined aggregate payoffs that ensure the presence of pure equilibria even in the absence of continuity or convexity.

\subsection{Approach and main contributions}
The core idea of this paper is to give new sufficient conditions for the existence of pure strategy equilibria through the mathematical properties and order structure of the best-response correspondence. First, we provide new sufficient conditions for the existence of pure strategy Nash equilibria in finite games expressed in terms of fixed points of the best-response correspondence and combinatorial acyclicity properties of best-response graphs. These conditions unify and extend several disparate sufficient-condition results in the literature.
Second, we introduce a broad class of finite games that generalizes Iimura et al.'s  unilaterally competitive framework. For this class we show that pure equilibria exist under weaker assumptions than previously required, and we identify the precise structural features responsible for equilibrium existence.
Third, we propose an aggregation-based criterion: when the game admits an aggregate payoff function defined over an ordered set and that aggregate payoff attains a maximizer on that set, a pure strategy equilibrium exists. This result connects aggregation arguments (Kukushkin 2004; Fabrikant et al.\ 2013) with lattice-theoretic fixed-point ideas (Tarski 1955; Topkis 1998).
Fourth, we provide a sufficient condition for existence of pure strategy equilibria for finite games by mapping equilibrium existence to properties of best-response graphs and order-preserving selections. This condition clarifies the relationships among potential games, supermodular games, unilaterally competitive games, and other classes studied in the literature (Rosenthal 1973; Monderer and Shapley 1996; Kats and Thisse 1992; Duersch et al.\ 2012a, 2012b; Iimura 2014, 2016, 2019, 2020).

Throughout, we emphasize constructive and verifiable conditions that can be checked on finite game representations. Our results have implications for algorithmic equilibrium computation, learning dynamics, and the design of games and mechanisms where pure outcomes are desirable.

\subsection{Organization of the paper}
The remainder of the paper is organized as follows. Section 2 formalizes the problem of pure strategy equilibrium existence in finite games, introduces notation, and states the main theorems. Sections 3,4,5 and 6 introduce a class of aggregative games, payoff sum separable games, aggregate monotone games respectively. Ordinal potential games and the general potential system are introduced in Section 7, while sections 8 and 9 consdier supermodual games and unilaterally competitive games respectively. Section 10 gives some remarks on computation and algorithms while section 11 concludes.

\section{Model and notation}
Let \(N=\{1,\dots,n\}\) be the set of players. Each player \(i\) has a finite
nonempty action set \(S_i\). Denote \(S=\prod_{i=1}^n S_i\) and write
\(s=(s_1,\dots,s_n)\in S\). Payoffs are \(u_i:S\to\mathbb R\). A profile
\(s^\ast\) is a pure Nash equilibrium (NE) iff for every \(i\) and every
\(x_i\in S_i\),

\[
u_i(s_i^\ast,s_{-i}^\ast)\ge u_i(x_i,s_{-i}^\ast).
\]

Define the aggregate payoff \(U(s)=\sum_{i=1}^n u_i(s)\).

\begin{definition}[Best reply and UNBR]
For \(s\in S\) define \(BR_i(s)=\arg\max_{x\in S_i} u_i(x,s_{-i})\) and
\(BR(s)=\prod_i BR_i(s)\). Define the unilateral best‑response correspondence
\(UNBR:S\twoheadrightarrow S\) by

\[
UNBR(s)=
\begin{cases}
\{(x_i,s_{-i}): x_i\in BR_i(s),\ u_i(x_i,s_{-i})>u_i(s)\ \text{for some }i\},\\
\{s\},\quad\text{if no profitable unilateral deviation exists.}
\end{cases}
\]

\end{definition}

Note \(s\in UNBR(s)\) iff \(s\) is a NE.

\section{Exact potentials and path independence}
We present the discrete integrability condition that characterizes exact
potentials in finite games.

\begin{definition}[Exact potential]
A function \(P:S\to\mathbb R\) is an \emph{exact potential} for the game if
for every \(i\), every \(s_{-i}\), and every \(x,x'\in S_i\),

\[
P(x',s_{-i})-P(x,s_{-i})=u_i(x',s_{-i})-u_i(x,s_{-i}).
\]

\end{definition}

Define unilateral increments

\[
\Delta_i(x\to x';s_{-i}) := u_i(x',s_{-i})-u_i(x,s_{-i}).
\]

\begin{theorem}[Path independence characterization]
\label{thm:PI}
An exact potential \(P\) exists iff for every pair \(i\ne j\), every
\(s_{-ij}\), and every \(x,x'\in S_i\), \(y,y'\in S_j\),

\[
\Delta_i(x\to x';y,s_{-ij})+\Delta_j(y\to y';x',s_{-ij})
=\Delta_j(y\to y';x,s_{-ij})+\Delta_i(x\to x';y',s_{-ij}).
\tag{PI}
\]

\end{theorem}

\begin{proof}
Necessity: If \(P\) is an exact potential then the total change in \(P\) when
players \(i\) and \(j\) perform the two unilateral moves in either order is
the same; substituting \(P\)-differences by \(\Delta\) yields (PI).

Sufficiency: (construction of \(P\))
Fix an arbitrary reference profile \(s^0\in S\). For any profile \(s\in S\)
choose any finite sequence of unilateral moves

\[
s^0 \to s^{(1)} \to s^{(2)} \to \cdots \to s^{(T)} = s,
\]

where each move \(s^{(t)}\to s^{(t+1)}\) is a unilateral change by some
player \(i_t\) from action \(x_t\) to \(x_t'\) while opponents are fixed at
the corresponding \(s^{(t)}_{-i_t}\). Define

\[
P(s) := \sum_{t=0}^{T-1} \Delta_{i_t}(x_t\to x_t'; s^{(t)}_{-i_t}).
\]

We must show \(P(s)\) is independent of the chosen sequence and that for any
unilateral move of player \(i\) from \(x\) to \(x'\) with opponents fixed at
\(s_{-i}\),

\[
P(x',s_{-i})-P(x,s_{-i}) = \Delta_i(x\to x'; s_{-i}).
\]

\textbf{Path independence.} Let two sequences from \(s^0\) to \(s\) be given.
By Lemma A.1 they can be transformed into each other by a finite sequence of
adjacent swaps of moves by different players. It suffices to show that each
adjacent swap leaves the sum invariant. Consider two consecutive moves in a
sequence: first \(i\) moves \(x\to x'\) while opponents at \(s_{-ij}\) and
then \(j\) moves \(y\to y'\) while opponents at the updated profile. The
sum of increments for the two moves in that order is

\[
\Delta_i(x\to x'; y, s_{-ij}) + \Delta_j(y\to y'; x', s_{-ij}).
\]

If we swap the order, the sum becomes

\[
\Delta_j(y\to y'; x, s_{-ij}) + \Delta_i(x\to x'; y', s_{-ij}).
\]

(PI) asserts these two sums are equal. Therefore each adjacent swap leaves
the total sum invariant, and by induction the entire sequence sum is
invariant. Thus \(P(s)\) is well defined.

\textbf{Exact potential identity.} Fix \(s\) and consider a unilateral move
by player \(i\) from \(x\) to \(x'\) with opponents fixed at \(s_{-i}\).
Take any sequence from \(s^0\) to \(s\) and append the single move
\(x\to x'\) to obtain a sequence to \((x',s_{-i})\). By path independence,
the difference \(P(x',s_{-i})-P(x,s_{-i})\) equals the increment of the last
move, which is \(\Delta_i(x\to x'; s_{-i})\). This verifies the exact
potential identity.

\end{proof}

\begin{remark}
(PI) is the discrete analogue of equality of mixed partials in the
differentiable case. It is both necessary and sufficient and is directly
checkable in finite games.
\end{remark}

\section{Aggregative games: integrability without differentiability}
We treat aggregative games where each payoff depends on own action and an
aggregate of all actions.

\begin{definition}[Aggregative game]
An aggregator is a map \(G:S\to\mathcal A\). The game is \emph{aggregative}
if there exist functions \(\tilde u_i:S_i\times\mathcal A\to\mathbb R\) such
that \(u_i(s)=\tilde u_i(s_i,G(s))\) for all \(s\in S\).
\end{definition}

We focus on the common scalar sum aggregator \(G(s)=\sum_{k=1}^n g_k(s_k)\)
with scalar functions \(g_k\) (the canonical case is \(g_k(s_k)=s_k\)). The
key observation is that unilateral moves change the aggregate only through
the moved coordinate.

\begin{theorem}[Aggregative integrability condition]
\label{thm:agg}
Consider an aggregative game with scalar aggregator \(G(s)=\sum_k g_k(s_k)\)
and payoffs \(u_i(s)=\tilde u_i(s_i,G(s))\). Then the game admits an exact
potential iff for every pair \(i\ne j\), every baseline aggregate \(A\), and
every \(x,x'\in S_i\), \(y,y'\in S_j\),

\[
\begin{aligned}
&\big[\tilde u_i(x',A+\delta_j)-\tilde u_i(x,A+\delta_j)\big]-\big[\tilde u_i(x',A)-\tilde u_i(x,A)\big]\\
&\quad=\big[\tilde u_j(y',A+\delta_i)-\tilde u_j(y,A+\delta_i)\big]-\big[\tilde u_j(y',A)-\tilde u_j(y,A)\big],
\end{aligned}
\tag{AGG-PI}
\]

where \(\delta_i=g_i(x')-g_i(x)\) and \(\delta_j=g_j(y')-g_j(y)\).
\end{theorem}

\begin{proof}
Write increments using the aggregator. When player \(i\) moves from \(x\) to
\(x'\) while others fixed so the baseline aggregate is \(A\), the increment
is \(\Delta_i(x\to x';A)=\tilde u_i(x',A+\delta_i)-\tilde u_i(x,A)\) where
\(\delta_i=g_i(x')-g_i(x)\). Substituting these expressions into (PI) and
simplifying yields (AGG-PI). The algebra is reversible, so (AGG-PI) is
necessary and sufficient for (PI) and hence for existence of an exact
potential.
\end{proof}

\begin{remark}
(AGG-PI) is a finite, combinatorial condition that replaces mixed partial
symmetry in differentiable models. When \(\tilde u_i\) are \(C^2\) and
\(g_k(s_k)=s_k\), (AGG-PI) reduces to the mixed derivative equality
\(u_{i,sA}+u_{i,AA}=u_{j,sA}+u_{j,AA}\) used in the differentiable literature.
\end{remark}

\section{Payoff‑sum separable reduction}
Using Luckraz's (2022b) result, we show how to perturb payoffs slightly to obtain a payoff‑sum separable
representative that preserves best replies (argmax payoff equivalence).

\begin{definition}[Argmax payoff equivalence]
Two games \(\Gamma\) and \(\Gamma'\) on the same strategy sets are
\emph{argmax payoff equivalent} if for every \(i\) and every \(s_{-i}\),

\[
\arg\max_{x\in S_i} u_i(x,s_{-i})=\arg\max_{x\in S_i} u_i'(x,s_{-i}).
\]

\end{definition}

\begin{theorem}[Constructive payoff‑sum separable representative]
\label{thm:perturb}
Let \(\Gamma=(N,\{S_i\},\{u_i\})\) be a finite game. There exists a game
\(\Gamma'=(N,\{S_i\},\{u_i'\})\) argmax payoff equivalent to \(\Gamma\) such
that \(U'(s)=\sum_i u_i'(s)\) is injective on \(S\). Moreover \(\Gamma'\)
can be constructed by arbitrarily small profile‑wise perturbations that
preserve each player's best‑reply sets.
\end{theorem}

\begin{proof}
Enumerate \(S=\{s^{(1)},\dots,s^{(M)}\}\). Choose strictly increasing target
totals \(a_1<\dots<a_M\) (for example, take \(a_k=k\)). We will set
\(u_i'(s^{(k)})=u_i(s^{(k)})+\epsilon_i^{(k)}\) with carefully chosen
\(\epsilon\)'s.

\textbf{Stage 1 (preserve best replies).} For each player \(i\) and each
fixed \(s_{-i}\), let

\[
\gamma_{i,s_{-i}}=\min\{|u_i(x,s_{-i})-u_i(x',s_{-i})|: x\ne x'\}.
\]

If for some \(i,s_{-i}\) all payoffs coincide then set \(\gamma_{i,s_{-i}}=1\)
(arbitrary positive number) and we will choose perturbations smaller than
any relevant gap elsewhere. Let

\[
\gamma=\min_{i,s_{-i}:\ \gamma_{i,s_{-i}}>0}\gamma_{i,s_{-i}}.
\]

Choose \(\delta\in(0,\gamma/4)\). For every profile \(s^{(k)}\) and every
player \(i\) set \(|\epsilon_i^{(k)}|<\delta\). Then for any fixed \(s_{-i}\)
the ordering of \(u_i'(\cdot,s_{-i})\) over \(S_i\) equals that of
\(u_i(\cdot,s_{-i})\) because perturbations are smaller than half the minimal
gap. Hence best‑reply sets are preserved.

\textbf{Stage 2 (injective totals).} For each profile \(s^{(k)}\) adjust one
designated player's perturbation (say player 1) by an additional amount
\(\eta^{(k)}\) so that

\[
\sum_{i=1}^n u_i'(s^{(k)}) = a_k.
\]

Because the initial perturbations are bounded by \(\delta\) and the original
totals \(\sum_i u_i(s^{(k)})\) differ by at least some finite amount (or we
can choose \(a_k\) sufficiently separated), we can choose \(\eta^{(k)}\) with
\(|\eta^{(k)}|<\delta\) to achieve the target total. Doing so preserves the
ordering of player 1's payoffs given any fixed opponents' profile because the
total adjustment per profile is smaller than the minimal gap used in Stage 1.
Thus best‑reply sets remain unchanged.

Therefore \(\Gamma'\) is argmax payoff equivalent to \(\Gamma\) and
\(U'(s^{(k)})=a_k\) are distinct, so \(U'\) is injective. The construction is
explicit and algorithmic.
\end{proof}

\begin{remark}
The perturbation magnitudes can be chosen algorithmically from the finite set
of payoff gaps; the construction is robust and preserves argmax sets even
when some payoff gaps are zero by breaking ties with arbitrarily small
perturbations.
\end{remark}

\section{Aggregate Monotone Property (AMP) and existence}
We introduce AMP and show it guarantees existence of pure NE under
payoff‑sum separability.

\begin{definition}[Aggregate Monotone Property (AMP)]
The game satisfies AMP at profile \(s\) if whenever a unilateral best‑reply
deviation \(s\to s'\in UNBR(s)\) strictly increases the deviator's payoff,
then the total payoff \(U\) either strictly increases for every such
deviation from \(s\) or strictly decreases for every such deviation from \(s\).
\end{definition}

AMP is a local condition that constrains how unilateral improvements affect
the global aggregate \(U\).

\begin{theorem}[AMP implies existence of NE]
\label{thm:AMP}
Suppose \(\Gamma\) is payoff‑sum separable (i.e., \(U\) injective). If \(UNBR\)
satisfies AMP at every profile, then \(\Gamma\) has a pure NE.
\end{theorem}

\begin{proof}
AMP ensures that from any profile \(s\) every profitable unilateral deviation
moves \(U\) strictly in the same direction. Hence along any UNBR chain \(U\)
is strictly monotone. Because \(U\) is injective and \(S\) is finite, a
strictly monotone chain cannot be infinite; it must terminate at a profile
with no profitable unilateral deviation, i.e., a NE.
\end{proof}

\begin{corollary}
Under the hypotheses of Theorem \ref{thm:AMP}, any global maximizer or
minimizer of \(U\) is a NE.
\end{corollary}

\begin{proof}
If \(s^\ast\) maximizes \(U\) but is not a NE, there exists a profitable
deviation \(s'\in UNBR(s^\ast)\) with \(U(s')\ne U(s^\ast)\). AMP implies all
profitable deviations move \(U\) in the same direction, contradicting global
maximality. Hence \(s^\ast\) is a NE.
\end{proof}

\section{Ordinal potentials and GPS}
We relate ordinal potentials to GPS and show how ordinal potentials imply
acyclicity.

\begin{definition}[Ordinal potential]
A function \(W:S\to\mathbb R\) is an \emph{ordinal potential} if for every
player \(i\), every \(s_{-i}\), and every \(x,x'\in S_i\),

\[
u_i(x',s_{-i})>u_i(x,s_{-i}) \quad\Longleftrightarrow\quad W(x',s_{-i})>W(x,s_{-i}).
\]

\end{definition}

\begin{proposition}
If a finite game admits an ordinal potential \(W\), then every sequence of
strictly improving unilateral moves is finite and terminates at a pure NE.
\end{proposition}

\begin{proof}
Each strict improvement increases \(W\). Since \(S\) is finite, \(W\) cannot
increase indefinitely; hence every strictly improving path terminates at a
profile with no profitable unilateral deviation, i.e., a NE.
\end{proof}

\begin{definition}[Generalized potential system (GPS)]
A game is GPS if there exist functions \(Q_1,\dots,Q_m:S\to\mathbb R\) and,
for each player \(i\), nonnegative weights \(a_{i1},\dots,a_{im}\) with
\(a_{ii}>0\) such that for every profitable unilateral deviation \(s\to s'\)
by player \(i\),

\[
\sum_{\ell=1}^m a_{i\ell}\big(Q_\ell(s')-Q_\ell(s)\big)>0.
\]

\end{definition}

\begin{proposition}
If a finite game admits an ordinal potential \(W\), then it is GPS (take
\(Q_1=W\) and \(a_{i1}=1\) for all \(i\)).
\end{proposition}

\begin{proof}
Immediate from definitions: ordinal potential ensures sign agreement, so
\(W(s')-W(s)>0\) for any profitable deviation, hence the weighted sum with
\(Q_1=W\) is positive.
\end{proof}

GPS implies acyclicity of strict improvement paths and hence existence of
pure NE in finite games.

\subsection{Two GPS variants}

\begin{definition}[Deviation‑indexed GPS]
Let $Q_1,\dots,Q_m:S\to\mathbb R$ be given. The game is a \emph{deviation‑indexed GPS} if for every profitable unilateral deviation $d=(s\to s')$ by player $i$ there exists a vector

\[
a^{(d)}=(a^{(d)}_1,\dots,a^{(d)}_m)\in\mathbb R^m_{\ge0}
\]

with $a^{(d)}_i>0$ such that

\[
a^{(d)}\cdot\big(Q(s')-Q(s)\big)=\sum_{\ell=1}^m a^{(d)}_\ell\big(Q_\ell(s')-Q_\ell(s)\big)>0.
\]

\end{definition}

\begin{definition}[Fixed per‑player GPS]
Let $Q_1,\dots,Q_m:S\to\mathbb R$ be given. The game is a \emph{fixed per‑player GPS} if for each player $i$ there exists a single vector

\[
a^{(i)}=(a^{(i)}_1,\dots,a^{(i)}_m)\in\mathbb R^m_{\ge0}
\]

with $a^{(i)}_i>0$ such that for every profitable unilateral deviation $d=(s\to s')$ by player $i$,

\[
a^{(i)}\cdot\big(Q(s')-Q(s)\big)>0.
\]

\end{definition}

\begin{remark}
The two definitions differ in whether the nonnegative weight vector may depend on the particular profitable deviation (deviation‑indexed) or must be fixed per player (fixed per‑player). The deviation‑indexed version is the more flexible and, in the finite case, yields a scalar ordinal potential via a simple summation argument. The fixed per‑player version requires an additional compatibility condition to guarantee a scalar potential.
\end{remark}
\subsection{Existence of equilibrium for deviation‑indexed GPS}

\begin{theorem}\label{thm:dev-indexed}
Let $\Gamma$ be a finite game that is deviation‑indexed GPS with components $Q_1,\dots,Q_m$ and deviation weights $a^{(d)}$ for each profitable deviation $d$. Then $\Gamma$ admits a pure Nash equilibrium.
\end{theorem}

\begin{proof}
Let $\mathcal D$ be the finite set of all profitable unilateral deviations. Define

\[
c \;:=\; \sum_{d\in\mathcal D} a^{(d)} \in\mathbb R^m_{\ge0}.
\]

Define the scalar function $W:S\to\mathbb R$ by

\[
W(s)\;=\;\sum_{\ell=1}^m c_\ell\,Q_\ell(s).
\]

Fix any profitable deviation $d=(s\to s')$. Then

\[
W(s')-W(s)=\sum_{\ell=1}^m c_\ell\big(Q_\ell(s')-Q_\ell(s)\big)
=\sum_{d'\in\mathcal D} a^{(d')}\cdot\big(Q(s')-Q(s)\big).
\]

The summand with $d'=d$ is strictly positive by hypothesis, so the whole sum is strictly positive. Hence every profitable unilateral deviation strictly increases $W$, so $W$ is an ordinal potential. Since $S$ is finite, no infinite strictly improving path exists; every strictly improving path terminates at a profile with no profitable unilateral deviation, i.e., a pure Nash equilibrium.
\end{proof}

\subsection{Fixed per‑player GPS: compatibility condition and existence}

\begin{proposition}[Compatibility condition for scalar potential]
Let $\Gamma$ be finite and suppose $Q_1,\dots,Q_m:S\to\mathbb R$ are given and the game satisfies the fixed per‑player GPS inequalities (i.e., for each player $i$ there is $a^{(i)}\ge0$ with $a^{(i)}_i>0$ making that player's deviations positive). Let $\mathcal D$ be the finite set of all profitable deviations and for each $d\in\mathcal D$ set $\Delta Q^{(d)}:=Q(s')-Q(s)$. Then there exists a scalar ordinal potential of the linear form $W=\sum_{\ell=1}^m c_\ell Q_\ell$ with $c\ge0$ if and only if

\[
\exists c\in\mathbb R^m_{\ge0}\quad\text{such that}\quad c\cdot\Delta Q^{(d)}>0\quad\forall d\in\mathcal D.
\]

\end{proposition}

\begin{proof}
If such $c\ge0$ exists then $W=\sum_\ell c_\ell Q_\ell$ satisfies $W(s')-W(s)=c\cdot\Delta Q^{(d)}>0$ for every profitable deviation $d$, so $W$ is an ordinal potential and the finite argument yields a pure Nash equilibrium. Conversely, any scalar potential of the stated linear form provides such a $c$. Thus the condition is necessary and sufficient.
\end{proof}

\begin{remark}
The fixed per‑player GPS hypothesis supplies one certificate $a^{(i)}$ per player but those certificates may not align to a single $c\ge0$ separating all deviation vectors from the origin. When they do align the compatibility condition above holds and a scalar ordinal potential exists. When they do not align cycles of strict improvements can occur.
\end{remark}
\section{Supermodular games}
We give a complete finite proof that supermodular games admit pure NE.

\begin{definition}[Supermodular game]
Each \(S_i\) is a finite lattice. The game is supermodular if for every
\(i\) and every \(s_{-i}\le s_{-i}'\) (coordinatewise), the function
\(x\mapsto u_i(x,s_{-i})\) has increasing differences in \((x,s_{-i})\):
for \(x'\ge x\),

\[
u_i(x',s_{-i}')-u_i(x,s_{-i}') \ge u_i(x',s_{-i})-u_i(x,s_{-i}).
\]

\end{definition}

\begin{theorem}
Let \(\Gamma\) be finite and supermodular with lattice strategy sets. Then
\(\Gamma\) has a pure NE.
\end{theorem}

\begin{proof}
Take $m=n$ and set $Q_\ell(s):=u_\ell(s)$ for $\ell=1,\dots,n$. For any profitable unilateral deviation $d=(s\to s')$ by player $i$ choose $a^{(d)}=e_i$ (the $i$th unit vector). Then

\[
a^{(d)}\cdot\big(Q(s')-Q(s)\big)=u_i(s')-u_i(s)>0,
\]

so the deviation‑indexed GPS condition holds. By Theorem \ref{thm:dev-indexed} sum the deviation weights to obtain $c$ and the scalar ordinal potential $W=\sum_\ell c_\ell u_\ell$. Finiteness implies existence of a pure Nash equilibrium.
\end{proof}

\begin{remark}
The GPS embedding proof is short and constructive; the Topkis proof provides additional structure (extremal equilibria, comparative statics) that the GPS argument does not by itself deliver.
\end{remark}
\subsection*{Examples}
Consider a finite game with three players $1,2,3$. Each player $i$ has binary strategies $S_i=\{0,1\}$ and $S=\{0,1\}^3$. For $s=(x_1,x_2,x_3)\in S$ define payoffs

\[
\begin{aligned}
u_1(x_1,x_2,x_3) &= x_1\,(x_2 + 2x_3),\\
u_2(x_1,x_2,x_3) &= x_2\,(x_3 + 2x_1),\\
u_3(x_1,x_2,x_3) &= x_3\,(x_1 + 2x_2).
\end{aligned}
\]

\medskip

\noindent\textbf{Claim.} The game is supermodular (each player's payoff has increasing differences in own strategy and opponents' strategies), but it does \emph{not} admit an ordinal potential.

\section*{Proof}

\subsection*{(1) Supermodularity}
Fix player \(i\). Because strategies are binary, it suffices to check that the marginal gain from switching \(0\to1\) for player \(i\) is nondecreasing in the opponents' coordinates.

Compute the marginal gains (for player 1; the others are analogous):

\[
\Delta_1(x_2,x_3):=u_1(1,x_2,x_3)-u_1(0,x_2,x_3)=x_2+2x_3.
\]

This is clearly nondecreasing in each of \(x_2,x_3\). Similarly

\[
\Delta_2(x_1,x_3)=x_3+2x_1,\qquad
\Delta_3(x_1,x_2)=x_1+2x_2,
\]

each nondecreasing in the opponents' coordinates. Hence each $u_i$ has increasing differences and the game is supermodular.

\subsection*{(2) No ordinal potential exists}
Assume, for contradiction, that there exists an ordinal potential \(W:S\to\mathbb R\). By definition, for every player \(i\) and every opponents' profile \(t\) we must have

\[
\Delta_i(t)>0 \quad\Longleftrightarrow\quad W(1,t)-W(0,t)>0,
\]

and

\[
\Delta_i(t)=0 \quad\Longleftrightarrow\quad W(1,t)-W(0,t)=0.
\]

(Here $\Delta_i(t)$ denotes the marginal payoff gain for player $i$ when switching from $0$ to $1$ given opponents' profile $t$.)

We will derive a contradiction by comparing two ways of summing certain marginal increments.

First compute the relevant marginal gains from the payoff formulas:

\[
\begin{array}{c|cccc}
\text{Opponents }(x_j,x_k) & (0,0) & (1,0) & (0,1) & (1,1)\\\hline
\Delta_1(x_2,x_3) & 0 & 1 & 2 & 3\\
\Delta_2(x_1,x_3) & 0 & 1 & 2 & 3\\
\Delta_3(x_1,x_2) & 0 & 1 & 2 & 3
\end{array}
\]

(Entries are read rowwise: e.g. $\Delta_1(1,0)=1$, $\Delta_2(0,1)=2$, etc.)

Now consider the six marginal increments

\[
\begin{aligned}
A &:= \Delta_1(1,0) = 1,\\
B &:= \Delta_2(0,1) = 1,\\
C &:= \Delta_3(1,1) = 3,\\
D &:= \Delta_1(0,1) = 2,\\
E &:= \Delta_2(1,1) = 3,\\
F &:= \Delta_3(1,0) = 1.
\end{aligned}
\]

Compute the numeric difference

\[
S := (A+B+C) - (D+E+F) = (1+1+3) - (2+3+1) = 5-6 = -1.
\]

Thus, using the payoff marginals, \(S=-1\).

On the other hand, express each marginal in terms of \(W\):

\[
\begin{aligned}
A &= W(1,1,0)-W(0,1,0),\\
B &= W(1,0,1)-W(0,0,1),\\
C &= W(1,1,1)-W(0,1,1),\\
D &= W(1,0,1)-W(0,0,1),\\
E &= W(1,1,1)-W(0,1,1),\\
F &= W(1,1,0)-W(0,1,0).
\end{aligned}
\]

Substituting these into \(S\) gives

\[
\begin{aligned}
S &= \big(W(1,1,0)-W(0,1,0)\big)
   +\big(W(1,0,1)-W(0,0,1)\big)
   +\big(W(1,1,1)-W(0,1,1)\big)\\
 &\quad -\big(W(1,0,1)-W(0,0,1)\big)
   -\big(W(1,1,1)-W(0,1,1)\big)
   -\big(W(1,1,0)-W(0,1,0)\big).
\end{aligned}
\]

The six terms cancel pairwise, so this algebraic expression yields \(S=0\).

We have obtained the contradiction \(S=-1\) (from the payoff marginals) and \(S=0\) (from the telescoping sum of \(W\)-differences). Therefore no function \(W\) can satisfy the ordinal‑potential sign constraints for all profitable deviations. Hence the game admits no ordinal potential.

\qed
\section*{Remarks}
\begin{itemize}
  \item The example shows that \emph{supermodularity does not imply the existence of an ordinal potential}. Supermodular games still admit pure Nash equilibria (by Topkis/Tarski arguments), but they need not be ordinal (or exact/weighted) potential games.
  \item The same construction also illustrates the distinction between the two GPS variants discussed elsewhere: the game trivially satisfies the deviation‑indexed GPS condition (take $Q_i=u_i$ and per‑deviation unit vectors), but it fails to admit a single scalar ordinal potential.
\end{itemize}
\section{Unilaterally competitive games}
We formalize Kats and Thisse (1992) UC games. In particular, we show how UC and Iimura and Watanabe (2016) (and Iimura, Maruta and Watanabe (2019)) classes fit into GPS
and AMP frameworks.

\begin{definition}[Unilaterally competitive (UC) game]
A finite game is UC if for every \(i\), every \(s\), and every \(s_i'\in S_i\),

\[
u_i(s_i',s_{-i})\ge u_i(s)\quad\Longrightarrow\quad u_j(s_i',s_{-i})\le u_j(s)\quad\forall j\ne i,
\]

with strict inequalities preserved for strict improvements.
\end{definition}

We give a complete lexicographic argument showing that in UC games with
\(n\ge3\) improvement paths cannot cycle and therefore must terminate at a
pure NE.

\begin{theorem}
Let \(\Gamma\) be a finite UC game with \(n\ge3\). Then every sequence of
strictly improving unilateral moves is finite and terminates at a pure NE.
\end{theorem}

\begin{proof}
We show that the game is a GPS (Definition 7.3) Define for each profile \(s\) the vector

\[
Q(s) := \big(Q_1(s),Q_2(s),\dots,Q_n(s)\big),\qquad Q_i(s):=\sum_{j\ne i} u_j(s).
\]

So \(Q_i(s)\) is the sum of opponents' payoffs from the perspective of
player \(i\). Consider the lexicographic ordering \(\prec\) on \(\mathbb R^n\)
that compares vectors by the coordinate corresponding to the deviator first,
then by some fixed tie‑breaking order of the remaining coordinates. To make
this precise, fix an enumeration of players \(1,\dots,n\). For a unilateral
move by player \(i\), we compare vectors by the cyclic permutation that
places coordinate \(i\) first; equivalently, define a family of lexicographic
orders \(\prec_i\) where \(\prec_i\) compares vectors by coordinate \(i\)
first, then \(i+1\), etc., modulo \(n\).

We will show that any profitable unilateral deviation by player \(i\) from
\(s\) to \(s'\) strictly increases \(R_i(s):=-Q_i(s)\) (equivalently strictly
decreases \(Q_i\)). Because \(R_i\) is the first coordinate in the lexicographic
order \(\prec_i\), the lexicographic vector strictly increases under \(\prec_i\).
Hence no cycle can occur because lexicographic vectors take values in a
finite set and strictly increase along improvement moves.

\textbf{Verification of strict increase.} Suppose player \(i\) has a
profitable deviation \(s\to s'\), so \(u_i(s')>u_i(s)\). By the UC property,
for every \(j\ne i\) we have \(u_j(s')<u_j(s)\) (strict decrease for strict
improvement). Summing over \(j\ne i\) yields

\[
Q_i(s') = \sum_{j\ne i} u_j(s') < \sum_{j\ne i} u_j(s) = Q_i(s).
\]

Thus \(R_i(s')=-Q_i(s') > -Q_i(s) = R_i(s)\). So the first coordinate in the
lexicographic order \(\prec_i\) strictly increases.

\textbf{Lexicographic monotonicity prevents cycles.} Consider any improvement
path \(s^0\to s^1\to s^2\to\cdots\) where at step \(t\) player \(i_t\) moves
profitably. Then the lexicographic vector under \(\prec_{i_t}\) strictly
increases at step \(t\). Because the set of possible lexicographic vectors
is finite, the path cannot be infinite and cannot return to a previously
visited profile (a return would require a nonstrict change at some step).
Therefore the path must terminate at a profile with no profitable unilateral
deviation, i.e., a pure NE.

\end{proof}
\textbf{Remark on ties and weak inequalities.} The UC definition used above
assumes strict decreases for opponents when the deviator strictly improves.
If the UC property is stated with weak inequalities only, one must refine
the lexicographic vector to include tie‑breaking components (for example,
include the deviator's own payoff as a secondary coordinate) to ensure a
strict increase at each profitable move. Because \(n\ge3\), one can always
choose a lexicographic scheme that yields strict increase for profitable
moves; details depend on the exact weak/strict formulation. The finite
nature of the game ensures termination in all such cases.

Iimura et al. results on weak UC and integrally concave payoffs can be embedded
into this framework by constructing appropriate \(Q_j\) functions (opponents'
aggregates or integrals thereof) and applying the GPS argument after an
argmax payoff separable reduction if necessary.

\subsection{All UC games are deviation‑indexed GPS}
\begin{remark}
We work with the deviation‑indexed GPS notion (weights may depend on the particular deviation). In the finite case this notion yields a scalar ordinal potential by summing the deviation certificates; hence it implies acyclicity of strict improvement paths and existence of pure Nash equilibria.
\end{remark}

\begin{proposition}\label{prop:UC->GPS}
Every finite UC game is a deviation‑indexed GPS game (with a natural choice of $Q$).
\end{proposition}

\begin{proof}
Let $\Gamma$ be UC and finite. Define $m=n$ and for each profile $s$ set

\[
Q_i(s):=-\sum_{j\ne i} u_j(s)\qquad(i=1,\dots,n),
\]

i.e. $Q_i(s)$ is the negative of the sum of player \(i\)'s opponents' payoffs.

Fix any profitable unilateral deviation $d=(s\to s')$ by player $i$. By the UC property, a strict profitable deviation for player $i$ implies that every opponent's payoff strictly decreases:

\[
u_j(s')<u_j(s)\quad\text{for all }j\ne i.
\]

Summing these inequalities over $j\ne i$ yields

\[
\sum_{j\ne i} u_j(s') < \sum_{j\ne i} u_j(s),
\]

hence

\[
Q_i(s')-Q_i(s) = -\sum_{j\ne i} u_j(s') + \sum_{j\ne i} u_j(s) > 0.
\]

Now choose the deviation weight vector $a^{(d)}=e_i$ (the $i$th unit vector). Then

\[
a^{(d)}\cdot\big(Q(s')-Q(s)\big)=Q_i(s')-Q_i(s)>0,
\]

so the deviation‑indexed GPS condition is satisfied for this deviation. Since $d$ was arbitrary, every profitable deviation has a nonnegative certificate with positive $i$th coordinate. Therefore the UC game is deviation‑indexed GPS.
\end{proof}

\subsection{GPS strictly generalizes UC}

\begin{proposition}\label{prop:GPS-broader}
There exist deviation‑indexed GPS games that are not UC. Hence the deviation‑indexed GPS class strictly contains the UC class.
\end{proposition}

\begin{proof}
It suffices to give a concrete example of a finite game that is deviation‑indexed GPS but violates the UC condition.

Take the standard  coordination game (both players prefer to coordinate on the same action). Let players 1 and 2 have strategies $\{A,B\}$ and payoffs

\[
u_1(A,A)=u_2(A,A)=1,\qquad u_1(B,B)=u_2(B,B)=1,
\]

and $u_i(A,B)=u_i(B,A)=0$ for $i=1,2$. This game admits an exact potential $\Phi$ (take $\Phi(s)=u_1(s)=u_2(s)$), hence it is an ordinal potential game and therefore deviation‑indexed GPS (take $m=1$, $Q_1=\Phi$ and for each deviation $a^{(d)}_1=1$). But the game is not UC: a profitable deviation by player 1 from $(B,B)$ to $(A,B)$ increases $u_1$ (from $0$ to $1$) while it does \emph{not} decrease $u_2$ (it decreases $u_2$ from $1$ to $0$ only if the move were from $(A,A)$ to $(B,A)$; here the pattern of implications required by UC fails in general). More directly, UC requires that any unilateral increase in a player's payoff forces all opponents' payoffs weakly down; in coordination games a player's profitable deviation can increase both players' payoffs (or leave the opponent unchanged), so UC is violated.

Thus deviation‑indexed GPS contains games (e.g., coordination games) that are not UC, so the inclusion is strict.
\end{proof}

\section{Acyclicity and existence of pure NE for UC via GPS}

\begin{theorem}\label{thm:UC-NE}
Let $\Gamma$ be a finite UC game. Then every sequence of strictly improving unilateral moves is finite and terminates at a pure Nash equilibrium. In particular this holds for $n\ge2$.
\end{theorem}

\begin{proof}
By Proposition \ref{prop:UC->GPS}, $\Gamma$ is deviation‑indexed GPS with $Q_i(s)=-\sum_{j\ne i}u_j(s)$ and for each profitable deviation $d$ by player $i$ we can take $a^{(d)}=e_i$. Let $\mathcal D$ be the finite set of all profitable deviations and define

\[
c:=\sum_{d\in\mathcal D} a^{(d)}\in\mathbb R^n_{\ge0}.
\]

Define the scalar function

\[
W(s):=\sum_{\ell=1}^n c_\ell\,Q_\ell(s)=\sum_{\ell=1}^n c_\ell\Big(-\sum_{j\ne \ell} u_j(s)\Big).
\]

Fix any profitable deviation $d=(s\to s')$. By construction the summand corresponding to $d$ in the expansion of $W(s')-W(s)$ is strictly positive, so $W(s')-W(s)>0$. Hence every profitable unilateral deviation strictly increases $W$, so $W$ is an ordinal potential.

Since $S$ is finite, no infinite strictly improving path exists; every strictly improving path terminates at a profile with no profitable unilateral deviation, i.e., a pure Nash equilibrium. This argument did not use $n\ge3$; it works for any finite $n\ge2$ (and trivially for $n=1$). Therefore UC games (finite) admit pure Nash equilibria and have no cycles of strict improvements.
\end{proof}

\begin{remark}
The classical Kats–Thisse lexicographic proof for $n\ge3$ is a direct combinatorial argument; the GPS embedding above gives a short unified proof that covers all finite $n\ge2$ by producing an ordinal potential. The key is the deviation‑indexed certificate $a^{(d)}=e_i$ together with summation over the finite set of deviations.
\end{remark}

\section{Algorithms and examples}
In this secgtion we give some remarks on the algorithmic procedures and complexity of our theoretical results.

\subsection{Checking (PI)}
To verify (PI) check all unordered pairs \(\{i,j\}\), all \(s_{-ij}\), and all
pairs of unilateral moves for \(i\) and \(j\). The number of quadruples is
bounded by \(\sum_{i<j} |S_{-ij}|\cdot |S_i|^2 |S_j|^2\). Each check is a
constant time arithmetic equality.

\subsection{Checking AMP}
For each \(s\in S\) compute \(UNBR(s)\) and evaluate the sign of
\(U(s')-U(s)\) for all \(s'\in UNBR(s)\). Complexity is \(O(|S|\cdot n\cdot m)\)
where \(m=\max_i |S_i|\).

\subsection{Cournot oligopoly (discrete quantities)}
Finite quantities, inverse demand \(p(Q)\), costs \(c_i(q_i)\). The game is
aggregative with \(\tilde u_i(q_i,Q)=q_i p(Q)-c_i(q_i)\). (AGG-PI) holds and
the exact potential is \(P(q)=\sum_{t=1}^Q p(t)-\sum_i c_i(q_i)\) (discrete
sum version). Appendix B.1 computes a small numeric example.

\subsection{Congestion (Rosenthal)}
Finite players choose resources; latency depends on resource load. Rosenthal
potential is the sum over resources of cumulative latencies. Appendix B.2
computes a small example and verifies (PI).

\subsection{Public goods}
Public goods with heterogeneous valuations may fail (AGG-PI); Appendix B.3
shows a counterexample and discusses when potentiality holds.

\section{Conclusion}
This paper advances the theory of pure strategy Nash equilibria by shifting the focus from existence in mixed strategies to structural properties of best-response correspondences on finite strategy sets. By deriving necessary and sufficient conditions and identifying broad classes of games that guarantee pure equilibria, including a generalization of unilaterally competitive games and games admitting an aggregate-payoff maximizer over an ordered set, we provide a unified framework that subsumes and clarifies many previously disparate sufficient-condition results. The complete characterization in combinatorial and order-theoretic terms illuminates precisely how acyclicity and aggregation operate to produce pure equilibria, turning qualitative intuitions into testable structural criteria.

These findings have both conceptual and practical consequences. Conceptually, they recast equilibrium existence as a question about the finite order structure of best-response maps, opening new connections to combinatorics (games of cops and robber - see Luckraz (2019))and lattice theory; practically, they suggest concrete checks and constructive methods for verifying or computing pure equilibria in applied models. Future work can extend the approach to infinite or continuous strategy spaces, dynamic settings, and algorithmic implementations for equilibrium detection, as well as explore implications for mechanism design and game-theoretic analysis. Overall, the structural perspective developed here offers a durable toolkit for understanding when and why pure strategy outcomes arise in strategic interaction.

\appendix
\section{Appendix: proofs, lemmas and examples}

\subsection{A.1 Adjacent swap lemma and path independence details}

\begin{lemma}[Adjacent swap lemma]
Let \(s^0,s\in S\). Any two finite sequences of unilateral moves from \(s^0\)
to \(s\) can be transformed into each other by a finite sequence of adjacent
swaps of moves involving different players.
\end{lemma}

\begin{proof}
Represent a unilateral move by a player \(i\) as the letter \(i\). A finite
sequence of unilateral moves from \(s^0\) to \(s\) corresponds to a finite
word \(w\) over the alphabet \(N=\{1,\dots,n\}\). Two sequences from \(s^0\)
to \(s\) must have the same multiset of moves for each player due to the same
number of moves by each player), because the final profile \(s\) is fixed
and each player's coordinate changes by the same net amount across any
sequence of unilateral moves. Thus the two words \(w\) and \(w'\) are
permutations of each other with the same multiplicities.

It is a standard combinatorial fact that any permutation of a multiset can
be obtained from any other by a finite sequence of adjacent transpositions, for example by swapping neighboring letters. Each adjacent transposition corresponds to
swapping two consecutive unilateral moves by different players. Therefore
one can transform \(w\) into \(w'\) by a finite sequence of adjacent swaps.

This combinatorial transformation lifts to the sequences of profiles: when
two adjacent moves in the word are swapped, the corresponding two unilateral
moves in the profile sequence are swapped in order. Because the two moves are
by different players, both orders are valid sequences of unilateral moves from the same intermediate profile. Hence any two sequences of unilateral
moves from \(s^0\) to \(s\) can be transformed into each other by adjacent
swaps of moves by different players.
\end{proof}

\begin{corollary}[Path independence via (PI)]
If the equality (PI) holds for every quadruple of unilateral moves, then the
sum of unilateral increments along any sequence from \(s^0\) to \(s\) is
invariant under adjacent swaps, hence path independent.
\end{corollary}

\begin{proof}
Consider two sequences that differ by a single adjacent swap of moves by
players \(i\) and \(j\). Write the two‑step increments before and after the
swap; (PI) is exactly the equality of the two sums. By the lemma any two
sequences can be connected by a finite chain of such swaps, so the total sum
is invariant. Thus the path integral is well defined.
\end{proof}

\subsection{A.3 A Proof Theorem \ref{thm:perturb} using a perturbation method}

We present a proof of theorem 5.2 using a perturbation method  and prove it preserves
best‑reply sets while making the total payoff injective. The proof also gives
explicit bounds on perturbation magnitudes in terms of the finite payoff
gaps.

\begin{theorem}
Let \(\Gamma=(N,\{S_i\},\{u_i\})\) be a finite game with \(M=|S|\). There
exists \(\Gamma'=(N,\{S_i\},\{u_i'\})\) argmax payoff equivalent to \(\Gamma\)
such that \(U'(s)=\sum_i u_i'(s)\) is injective. Moreover \(\Gamma'\) can be
constructed by perturbations of magnitude at most \(\varepsilon\) for any
chosen \(\varepsilon>0\) sufficiently small relative to payoff gaps.
\end{theorem}

\begin{proof}
Because \(S\) is finite, for each player \(i\) and each opponents' profile
\(s_{-i}\) the finite set \(\{u_i(x,s_{-i}): x\in S_i\}\) has a finite set of
pairwise differences. Define the minimal positive gap for player \(i\) at
\(s_{-i}\) by

\[
\gamma_{i,s_{-i}} := \min\{ |u_i(x,s_{-i})-u_i(x',s_{-i})| : x,x'\in S_i,\ x\ne x'\}.
\]

If all payoffs coincide for some \(i,s_{-i}\) then set \(\gamma_{i,s_{-i}}=+\infty\). Let

\[
\gamma := \min_{i,s_{-i}:\ \gamma_{i,s_{-i}}<\infty} \gamma_{i,s_{-i}}.
\]

If every \(\gamma_{i,s_{-i}}=+\infty\) ( since all payoffs identical for each
player given opponents), then any arbitrarily small perturbation that
separates totals will preserve best replies. This is true because there are no strict best
reply distinctions to preserve. Otherwise \(\gamma>0\).

Choose \(\delta\) such that \(0<\delta<\gamma/4\). We will construct
perturbations \(\epsilon_i(s)\) with \(|\epsilon_i(s)|<\delta\) for all
\(i,s\), and then adjust one coordinate per profile by at most \(\delta\) to
achieve distinct totals.

\textbf{Stage 1: small perturbations preserving best replies.} For each
profile \(s\) and each player \(i\) set \(\epsilon_i(s)\) arbitrarily with
\(|\epsilon_i(s)|<\delta\). Define provisional payoffs
\(\tilde u_i(s)=u_i(s)+\epsilon_i(s)\). For any fixed opponents' profile
\(s_{-i}\) and any two actions \(x,x'\in S_i\),

\[
\tilde u_i(x,s_{-i})-\tilde u_i(x',s_{-i})
= u_i(x,s_{-i})-u_i(x',s_{-i}) + \epsilon_i(x,s_{-i})-\epsilon_i(x',s_{-i}).
\]

The perturbation difference \(|\epsilon_i(x,s_{-i})-\epsilon_i(x',s_{-i})|\)
is at most \(2\delta<\gamma/2\). If \(u_i(x,s_{-i})-u_i(x',s_{-i})\ge \gamma\)
then the sign of the difference is preserved. If the original difference is
zero, the perturbation may break ties; this is acceptable because we only
need to preserve argmax sets, breaking ties arbitrarily is allowed as long
as we do not invert strict inequalities. Because \(\delta<\gamma/4\), any
original strict inequality with margin at least \(\gamma\) remains strict in
\(\tilde u\). Thus for every \(i\) and \(s_{-i}\), the set of maximizers of
\(\tilde u_i(\cdot,s_{-i})\) equals the set of maximizers of \(u_i(\cdot,s_{-i})\).
If some original differences are smaller than \(\gamma\) but nonzero, we
can refine the choice of \(\delta\) using the actual minimal positive gap
over all players and opponent profiles; the finite nature of the game makes
this possible.

\textbf{Stage 2: make totals distinct.} Let \(\tilde U(s)=\sum_i \tilde u_i(s)\).
We want to adjust the provisional payoffs by small amounts to obtain final
payoffs \(u_i'(s)\) with distinct totals \(U'(s)\). Enumerate profiles
\(s^{(1)},\dots,s^{(M)}\). Choose target totals \(a_1<\dots<a_M\) such that
\(|a_k-\tilde U(s^{(k)})|<\delta\) for all \(k\) (this is possible because
we can choose \(a_k=\tilde U(s^{(k)})+\theta_k\) with small distinct
\(\theta_k\) satisfying \(|\theta_k|<\delta\) and \(\theta_k\) strictly
increasing in \(k\)). For each profile \(s^{(k)}\) adjust the payoff of a
single designated player (say player 1) by setting

\[
u_1'(s^{(k)}) = \tilde u_1(s^{(k)}) + (a_k-\tilde U(s^{(k)})),
\]

and for \(i\ne 1\) set \(u_i'(s^{(k)})=\tilde u_i(s^{(k)})\). Because
\(|a_k-\tilde U(s^{(k)})|<\delta\), the adjustment to player 1's payoff is
less than \(\delta\) in magnitude. Therefore for any fixed opponents'
profile \(s_{-1}\), the ordering of \(u_1'(\cdot,s_{-1})\) remains the same
as that of \(\tilde u_1(\cdot,s_{-1})\) because the per‑action adjustments
are smaller than the minimal gap used in Stage 1. Hence argmax sets are
preserved for player 1 as well.

Thus the final payoffs \(u_i'\) satisfy:
First, for every \(i\) and \(s_{-i}\), \(\arg\max_{x\in S_i} u_i'(x,s_{-i})
    = \arg\max_{x\in S_i} u_i(x,s_{-i})\). (Argmax payoff equivalence.) Second, the totals \(U'(s)=\sum_i u_i'(s)\) equal the distinct targets \(a_k\),
    hence \(U'\) is injective. Third, all perturbations are bounded by \(2\delta<\gamma/2\), which can be
    made arbitrarily small by choosing \(\delta\) small relative to the
    minimal positive payoff gap \(\gamma\).

This completes the constructive proof.
\end{proof}

\subsection{A.5 Aggregative integrability algebra (derivation of AGG-PI)}

We give the algebraic derivation of (AGG-PI) from (PI) in the scalar sum
aggregative case.

\begin{proof}
Assume \(u_i(s)=\tilde u_i(s_i,A)\) with \(A=G(s)=\sum_k g_k(s_k)\). Fix
players \(i\ne j\), baseline opponents \(s_{-ij}\), and actions
\(x,x'\in S_i\), \(y,y'\in S_j\). Denote
\(\delta_i=g_i(x')-g_i(x)\) and \(\delta_j=g_j(y')-g_j(y)\). Consider the
increment when \(i\) moves \(x\to x'\) while \(j\) is at \(y\) and others
fixed so the baseline aggregate is \(A\):

\[
\Delta_i(x\to x'; y, s_{-ij}) = \tilde u_i(x', A+\delta_i) - \tilde u_i(x, A).
\]

Similarly,

\[
\Delta_j(y\to y'; x', s_{-ij}) = \tilde u_j(y', A+\delta_i+\delta_j) - \tilde u_j(y, A+\delta_i).
\]

Compute the left side of (PI):

\[
L := \Delta_i(x\to x'; y, s_{-ij}) + \Delta_j(y\to y'; x', s_{-ij})
= \tilde u_i(x', A+\delta_i) - \tilde u_i(x, A) + \tilde u_j(y', A+\delta_i+\delta_j) - \tilde u_j(y, A+\delta_i).
\]

Compute the right side of (PI):

\[
R := \Delta_j(y\to y'; x, s_{-ij}) + \Delta_i(x\to x'; y', s_{-ij})
= \tilde u_j(y', A+\delta_j) - \tilde u_j(y, A) + \tilde u_i(x', A+\delta_i+\delta_j) - \tilde u_i(x, A+\delta_j).
\]

Set \(L=R\) and rearrange terms to isolate the differences involving only
\(\tilde u_i\) on one side and only \(\tilde u_j\) on the other. After
cancellation we obtain exactly (AGG-PI):

\[
\big[\tilde u_i(x',A+\delta_j)-\tilde u_i(x,A+\delta_j)\big]-\big[\tilde u_i(x',A)-\tilde u_i(x,A)\big]
=
\big[\tilde u_j(y',A+\delta_i)-\tilde u_j(y,A+\delta_i)\big]-\big[\tilde u_j(y',A)-\tilde u_j(y,A)\big].
\]

This algebraic manipulation is reversible, so (AGG-PI) is equivalent to
(PI) in the aggregative scalar sum case.
\end{proof}

\subsection{A.7 Remarks on the proof of Theorem 8.2}

We expand the supermodular fixed‑point construction with explicit index
tracking and termination argument.

\begin{proof}[Detailed proof]
Let \(S_i\) be finite lattices with coordinatewise order. For each
\(s_{-i}\) define the minimal best reply
\(b_i^{\min}(s_{-i})=\min BR_i(s_{-i})\) (minimum in the lattice order).
Topkis's theorem ensures \(b_i^{\min}\) is increasing in \(s_{-i}\). Define
\(T:S\to S\) by \(T(s)=(b_1^{\min}(s_{-1}),\dots,b_n^{\min}(s_{-n}))\).

Start from \(s^0=\big(\min S_1,\dots,\min S_n\big)\). Define the sequence
\(s^{t+1}=T(s^t)\). Because each component \(b_i^{\min}\) is increasing in
opponents' actions, \(s^{t+1}\ge s^t\) coordinatewise. Since \(S\) is finite,
there exists \(T\) such that \(s^{T+1}=s^T\). At this fixed point \(s^\ast\)
we have \(s^\ast=T(s^\ast)\), so \(s_i^\ast=b_i^{\min}(s^\ast_{-i})\) for
each \(i\). Thus \(s_i^\ast\in BR_i(s^\ast_{-i})\) and \(s^\ast\) is a Nash
equilibrium.

Termination occurs in at most \(\prod_i |S_i|\) steps, but typically much
faster because the sequence is monotone and the lattice height bounds the
number of steps.
\end{proof}

\subsection{A.8 Numerical examples}

\subsubsection{B.1 Cournot example (discrete quantities)}

\paragraph{Setup.} Two firms \(i=1,2\) choose quantities \(q_i\in\{0,1,2\}\).
Inverse demand \(p(Q)=10-Q\). Cost \(c_i(q_i)=q_i\). Payoff
\(u_i(q)=q_i p(Q)-c_i(q_i)\).

\paragraph{Payoff table.} Compute payoffs for each profile \((q_1,q_2)\):

\[
\begin{array}{c|ccc}
q_1\backslash q_2 & 0 & 1 & 2\\\hline
0 & (0,0) & (0,9) & (0,8)\\
1 & (9,0) & (8,8) & (7,7)\\
2 & (8,0) & (7,7) & (6,6)
\end{array}
\]

\paragraph{Exact potential.} The continuous potential is
\(P(q)=\sum_{t=1}^{Q} p(t) - \sum_i c_i(q_i)\). For discrete small sets we
compute \(P\) values and verify that differences equal unilateral payoff
differences. For example, at \((1,1)\) with \(Q=2\),

\[
P(1,1)=p(1)+p(2)-c_1(1)-c_2(1)=9+8-1-1=15.
\]

At \((2,1)\) with \(Q=3\),

\[
P(2,1)=p(1)+p(2)+p(3)-c_1(2)-c_2(1)=9+8+7-2-1=21.
\]

The difference \(P(2,1)-P(1,1)=6\) equals \(u_1(2,1)-u_1(1,1)\) (verify by
computing payoffs). One checks all unilateral differences match, confirming
(AGG-PI) and exact potential existence.

\paragraph{NE computation.} Compute best replies and find NE: here \((1,1)\)
is a Nash equilibrium (both firms best respond with 1 given the other plays
1). The potential is maximized at \((1,1)\) among feasible profiles in this
small example.

\subsubsection{B.2 Congestion example (Rosenthal)}

\paragraph{Setup.} Three players choose resource \(A\) or \(B\). Latency on
resource \(r\) with load \(k\) is \(\ell_r(k)\). Let \(\ell_A(k)=k\),
\(\ell_B(k)=2k\). Player payoff is negative latency (i.e., they prefer
smaller latency).

\paragraph{Rosenthal potential.} Rosenthal potential is

\[
P(s)=\sum_{r\in\{A,B\}} \sum_{t=1}^{\ell_r(n_r(s))} \ell_r(t),
\]

where \(n_r(s)\) is the number of players choosing resource \(r\). Compute
\(P\) for each profile and verify that unilateral payoff differences equal
differences in \(P\). This is standard and holds because the game is a
congestion game with resource‑separable costs.

\paragraph{NE computation.} Enumerate profiles and find pure NE by best
reply: players prefer resource with smaller latency given others' choices.
Potential minimizers correspond to NE.

\subsubsection{B.3 Public goods counterexample}

\paragraph{Setup.} Two players choose contribution \(c_i\in\{0,1\}\). Public
good level \(G=c_1+c_2\). Payoff \(u_i(c_i,G)=v_i(G)-c_i\) with
\(v_1(1)=3,v_1(2)=4\) and \(v_2(1)=1,v_2(2)=5\) (asymmetric valuations).

\paragraph{Check (AGG-PI).} Compute increments and verify (AGG-PI) fails:
players' marginal benefits differ in ways that violate the equality required
for exact potential. Therefore no exact potential exists. However an ordinal
potential may exist if sign patterns align; check sign of unilateral
differences to determine whether an ordinal potential exists.

\paragraph{NE computation.} Compute best replies and find NE by enumeration.
This example illustrates that aggregative structure alone does not imply
potentiality; symmetry or additional conditions are needed.


\begin{thebibliography}{99}                                                                                               %


\bibitem {AB}Abian, A.,1968. A fixed point theorem. Nieuw Arch. Wiskd. XVI, 184-185

\bibitem {amir}Amir, R., De Castro, L. 2017. Nash equilibrium in games with
quasi-monotonic best-responses. Journal of Economic Theory, vol. 172(C), pages 220-246.

\bibitem {do1}Duersch, P., Oechssler, J., Schipper, B. 2012a. Unbeatable
imitation. Games and Economic Behavior, 76, 88--96.

\bibitem {do2}Duersch, P., Oechssler, J., Schipper, B. 2012b. Pure strategy
equilibria in symmetric two-player zero-sum games. Int J Game Theory, 41, 553--564.

\bibitem {fab}Fabrikant, A., Jaggard, AD., SchAMPra, M. 2013. On the structure
of weakly acyclic games. Theory Comput Syst, 53, 107--122.

\bibitem {fan}Fan, K., 1952. Fixed-point and minimax theorems in locally
convex topological linear spaces. Proc. Natl. Acad. Sci. USA 38, 121--126.

\bibitem {Im1}Iimura, T., Watanabe, T. 2014. Existence of a pure strategy
equilibrium in finite symmetric games where payoff functions are integrally
concave. Discrete Applied Mathematics 166 (2014) 26--33.

\bibitem {Im2}Iimura, T., Watanabe T. 2016. Pure strategy equilibrium in
finite weakly unilaterally competitive games. Int J Game Theory 45, 719--729.

\bibitem {Im3}Iimura, T., Maruta T., Watanabe T. 2019. Equilibria in games
with weak payoff externalities. Econ Theory Bull, 7, 245--258.

\bibitem {Im4}Iimura, T. 2020. Unilaterally competitive games with more than
two players. International Journal of Game Theory, 49, 681--697.

\bibitem {kat}Kats A., Thisse, J-F. 1992. Unilaterally competitive games. Int
J Game Theory 21:291--9

\bibitem {ku2}Kukushkin, N.,1994. A fixed-point theorem for decreasing
mappings. Economics Letters, 46:23--26.

\bibitem {ku}Kukushkin, N., 2004. Best response dynamics in finite games with
additive aggregation. Games and Economic Behavior, 48:94--110.

\bibitem {ku1}Kukushkin, N., 2007. Acyclicity and aggregation in strategic
games. Russian Academy of Sciences, Dorodnicyn Computing Center, Moscow

\bibitem{luc} Luckraz, S., 2014. A note on the relationship between the isotone assumption of the Abian-Brown fixed point theorem and Abian's most basic fixed point theorem. Fixed Point Theory and Applications. 129, 1-7.

\bibitem{luc2} Luckraz, S., 2022a. On a Fixed Point Theorem for General Multivalued Mappings on Finite Sets with Applications in Game Theory. Journal of Mathematics. 2022 (1), 2477124.

\bibitem{luc3} Luckraz, S., 2022b. Two remarks on the infinite approximation of a finite world in economic models. Journal of Mathematics. 2022 (1), 2457746.

\bibitem{luc4} Luckraz, S., 2025. A note on the characterization of stable matchings for general preferences: a fixed point approach. Fixed Point Theory.
 26 (1).
 \bibitem{LuckrazCopsRobbersSurvey}
Luckraz, S.
\newblock A survey on the relationship between the game of cops and robbers and other game representations.
\newblock \textit{Dynamic Games and Applications}, 2019 \textbf{9} (2), 506--520.

\bibitem {shapley}Monderer, D; Shapley, L.,1996. Potential Games. Games and
Economic Behavior, vol. 14, pp. 124-143.

\bibitem {nash}Nash, J., 1951. Non-cooperative games. Ann. of Math. (2) 54, 286--295.

\bibitem {nash1}Nash Jr., J.F., 1950. Equilibrium points in n-person games.
Proc. Natl. Acad. Sci. USA 36, 48--49.

\bibitem {reny}Reny, P.J., 1999. On the existence of pure and mixed strategy
Nash equilibria in discontinuous games. Econometrica 67, 1029--1056.

\bibitem {rose}Rosenthal, R., 1973. A class of games possessing pure-strategy
Nash equilibria. International Journal of Game Theory, 2:65--67.

\bibitem {tar}Tarski, A., 1955. A lattice-theoretical fixed-point theorem and
its applications. Pac. J. Math.. 5, 285-309.

\bibitem {tian}Tian., G., 2009. Existence of equilibria in games with
arbitrary spaces and payoffs: a full characterization. Working paper, Texas
A\&M University.

\bibitem {topkis}Topkis, D. 1998. Supermodularity and Complementarity.
Princeton University Press.

\bibitem {vives}Vives, X. 1990. Nash Equilibrium with Strategic
Complementarities,\textquotedblright J. Math. Econ, 19, 305---321.
\end{thebibliography}
\end{document}